\def\year{2022}\relax
\documentclass[letterpaper]{article} 
\usepackage{aaai22}  
\usepackage{times}  
\usepackage{helvet}  
\usepackage{courier}  
\usepackage[hyphens]{url}  
\usepackage{graphicx} 
\usepackage{natbib}  
\usepackage{caption} 
\DeclareCaptionStyle{ruled}{labelfont=normalfont,labelsep=colon,strut=off} 
\usepackage{multicol}

\usepackage{algorithm}
\usepackage{algorithmic}

\usepackage{newfloat}
\usepackage{listings}
\floatstyle{ruled}
\newfloat{listing}{tb}{lst}{}
\floatname{listing}{Listing}
\usepackage{cite}
\usepackage{amsmath,amssymb,amsfonts}
\usepackage{algorithmic}
\usepackage{graphicx}
\usepackage{textcomp}
\usepackage{xcolor}

\usepackage{csquotes}
\usepackage{url}
\usepackage{mathtools}
\usepackage{csquotes}
\usepackage{enumitem}
\usepackage{xspace}
\usepackage{algorithmic}
\usepackage{algorithm}
\usepackage{subcaption}
\usepackage{array}
\usepackage{comment}
\usepackage[group-separator={,}]{siunitx}
\usepackage{bm}
\usepackage{placeins}
\usepackage{amsthm}
\usepackage{xcolor}
\usepackage{soul}
\usepackage{multirow}
\usepackage{makecell}
\usepackage{wrapfig}
\usepackage{booktabs}
\usepackage{hyperref}
\usepackage{multibib}

\def\BibTeX{{\rm B\kern-.05em{\sc i\kern-.025em b}\kern-.08em
    T\kern-.1667em\lower.7ex\hbox{E}\kern-.125emX}}

\newtheorem{definition}{Definition}
\newtheorem{theorem}{Theorem}
\newtheorem{lemma}[theorem]{Lemma}
\DeclareMathOperator*{\argmin}{\arg\min}   
\newcites{apx}{References}
\usepackage{silence}
\title{Meta-Learning for Online Update of Recommender Systems}
\author{
Minseok Kim\textsuperscript{\rm 1},
Hwanjun Song\textsuperscript{\rm 2},
Yooju Shin\textsuperscript{\rm 1},
Dongmin Park\textsuperscript{\rm 1},
Kijung Shin\textsuperscript{\rm 1},
Jae-Gil Lee\textsuperscript{\rm 1}\thanks{Jae-Gil Lee is the corresponding author.}}
\affiliations{
    \textsuperscript{\rm 1}KAIST,
    \textsuperscript{\rm 2}NAVER AI Lab,\\
    {\{minseokkim, yooju.shin, dongminpark, kijungs, jaegil\}}@kaist.ac.kr, \\
    hwanjun.song@navercorp.com
}

\begin{document}

\newcommand{\thickbar}[1]{\mathbf{\bar{\text{$#1$}}}}

\DeclarePairedDelimiter\ceil{\lceil}{\rceil}
\DeclarePairedDelimiter\floor{\lfloor}{\rfloor}

\renewcommand{\algorithmicrequire}{\textsc{Input:}}
\renewcommand{\algorithmicensure}{\textsc{Output:}}
\renewcommand{\algorithmiccomment}[1]{/*~#1~*/}
\newcommand{\algname}{{MeLON}}

\newcommand\kijung[1]{\textcolor{red}{[Kijung: #1]}}
\newcommand\minseok[1]{\textcolor{purple}{#1}}
\newcommand\change[1]{\textcolor{magenta}{#1}}
\newcommand\add[1]{\textcolor{blue}{#1}}
\newcommand\red[1]{\textcolor{red}{#1}}
\maketitle
\begin{abstract}
Online recommender systems should be always aligned with users' current interest to accurately suggest items that each user would like. Since user interest usually \emph{evolves} over time, the update strategy should be \emph{flexible} to quickly catch users' current interest from continuously generated new user-item interactions. Existing update strategies focus either on the importance of each user-item interaction or the learning rate for each recommender parameter, but such \emph{one}-directional flexibility is insufficient to adapt to varying relationships between interactions and parameters.
In this paper, we propose \algname{}, a meta-learning based novel online recommender update strategy that supports \emph{two}-directional flexibility. It is featured with an \emph{adaptive} learning rate for each \emph{parameter-interaction pair} for inducing a recommender to quickly learn users' up-to-date interest.
The procedure of \algname{} is optimized following a meta-learning approach: it learns how a recommender learns to generate the optimal learning rates for future updates.
Specifically, \algname{} first enriches the meaning of each interaction based on previous interactions and identifies the role of each parameter for the interaction; and then combines these two pieces of information to generate an adaptive learning rate.
Theoretical analysis and extensive evaluation on three real-world online recommender datasets validate the effectiveness of \algname{}.
\end{abstract}

\vspace*{-0.5cm}
\section{Introduction}
\label{sec:introduction}
The widespread of mobile devices enables a large number of users to connect to a variety of online services, such as video streaming\,\cite{davidson2010youtube}, shopping\,\cite{linden2003amazon}, and news\,\cite{gulla2017adressa}, where each user seeks only a few items out of a myriad of items in services.
To keep users involved, online services struggle to meet each user's needs \emph{accurately} by deploying 
personalized recommender systems\,\cite{koren2008factorization}, which suggest the items potentially interesting to him/her. 
In an online setting where a user's \emph{current} interest changes constantly, the online recommender should catch up each user's up-to-date interest to prevent its service from being stale\,\cite{he2016fast}.
To this end, recommender models are updated continuously in response to new user-item interactions.

In modern online recommender systems, \emph{fine-tuning} has been widely employed to update models since it is infeasible to re-train the models from scratch whenever new user-item interactions come in. 
Specifically, pre-trained models\,(i.e., snapshots trained on past user-item interactions) are fine-tuned based \emph{only} on new user-item interactions. 
Fine-tuning not only requires less computational cost but also has sufficient capability to reflect up-to-date information\,\cite{zhang2020retrain}. 
However, because \emph{few-shot} incoming user-item interactions are very sparse in the user-item domain\,\cite{finn2019online},
the standard fine-tuning scheme would not suit online recommender systems to \emph{quickly} adapt to up-to-date user interest.
Therefore, the key challenge is to overcome this data sparsity for fine-tuning.

To cope with this challenge, previous researches have been actively studied in two directions.
\begin{itemize}[leftmargin=9pt, noitemsep]
\item \textbf{Importance reweighting} adjusts the importance of \emph{each} new user-item interaction\,\cite{he2016fast, shu2019meta}. 
These methods receive the \emph{loss} of a new user-item interaction from a recommender as a supervisory signal and then determine how much the recommender should be fine-tuned by each user-item interaction.
\item \textbf{Meta-optimization} controls how much \emph{each} recommender parameter should be fine-tuned from new user-item interactions\,\cite{zhang2020retrain}. 
These methods determine a parameter-wise optimization strategy such as a learning rate, given the loss of the new user-item interactions and each parameter's previous optimization history.
\end{itemize}

\begin{figure}[t!]
\begin{center}
\includegraphics[width=0.47\textwidth]{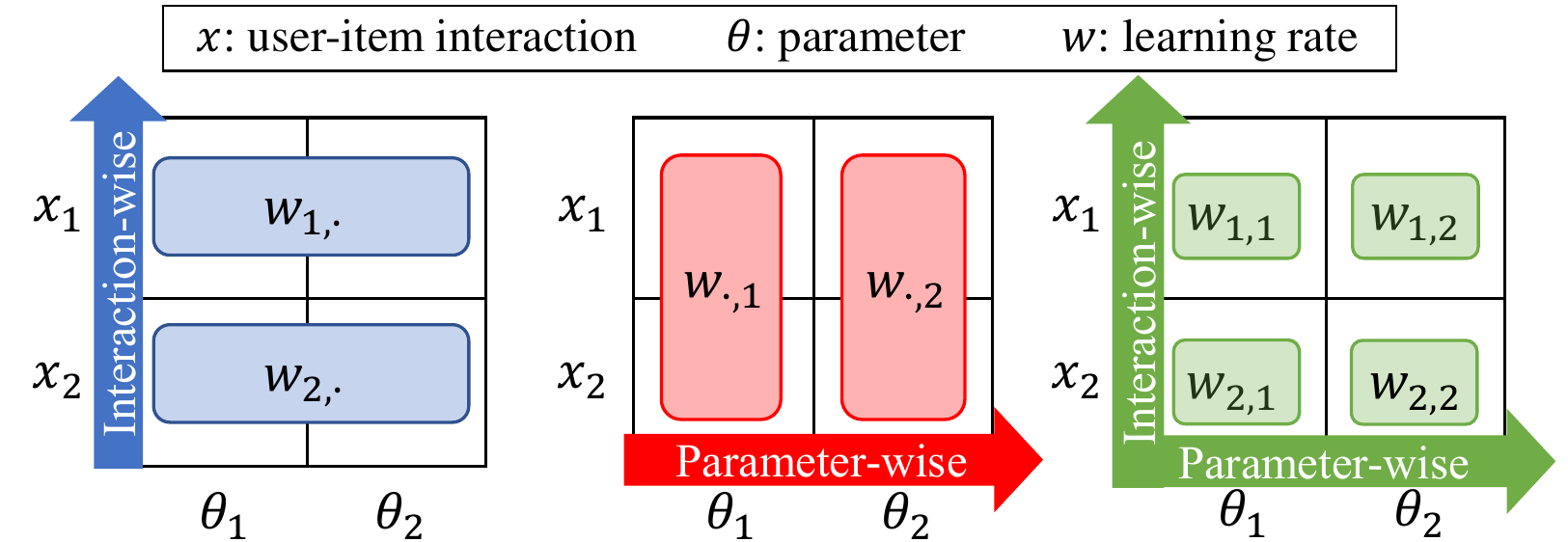}
\end{center}
\vspace*{-0.3cm}
\hspace*{0.05cm} {\scriptsize (a) Importance Reweighting.} \hspace*{0.2cm}{\scriptsize (b) Meta-Optimization.} \hspace*{0.45cm} {\scriptsize (c) Ours: \algname{}.}
\vspace*{-0.15cm}
\caption{Flexibility comparison of adjusting a learning rate $w$ to update a parameter $\theta$ given a user-item interaction $x$. While (a) importance reweighting and (b) meta-optimization support only one of the two learning perspectives, (c) \algname{} supports both of them.} 
\label{fig:motivating_example}
\vspace*{-0.6cm}
\end{figure}

\begin{figure*}[htb!]
\begin{center}
\includegraphics[width=1\textwidth]{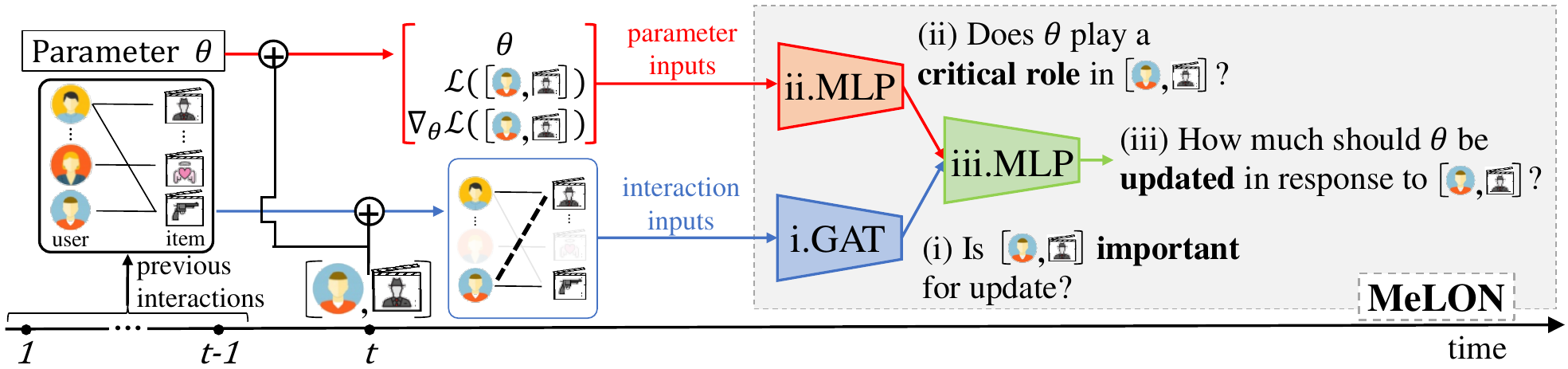}
\end{center}
\vspace*{-0.5cm}
\caption{Illustration of \algname{}'s procedure. \algname{} (i) represents the importance of a given user-item interaction for the current update based on previous interactions, (ii) identifies the role of each parameter for the interaction, and (iii) adapts a learning rate specific to each parameter-interaction pair.}
\label{fig:MeLON_idea_procedure}
\vspace*{-0.5cm}
\end{figure*}

These two orthogonal approaches focus on different aspects on learning: importance reweighting focuses on the impact of each \emph{user-item interaction}, as shown in Figure \ref{fig:motivating_example}(a), while meta-optimization focuses on that of each \emph{parameter}, as shown in Figure \ref{fig:motivating_example}(b).
That is, both approaches support the \emph{one}-directional flexibility in the either data or parameter perspective.
If we regard fine-tuning at each time period as a distinct task\,\cite{zhang2020retrain}, the role of each parameter in a recommender varies for different user-item interactions, because it is well known that an explanatory factor\,(i.e., parameter) in a neural network has different relevance toward different tasks\,\cite{bengio2013representation}. 
Thus, we contend that the flexibility in \emph{both} data and parameter perspectives should be achieved.
However, the two existing approaches lack this level of flexibility, possibly leading to sub-optimal recommendation quality.

In this paper, we propose, \textbf{\algname{}} (\underline{Me}ta-\underline{L}earning for \underline{ON}line recommender update), a novel online recommender update strategy that supports the flexibility in \emph{both} data and parameter perspectives.
It learns to \emph{adaptively} adjust the learning rate of each parameter for each new user-item interaction, as illustrated in Figure \ref{fig:motivating_example}(c).
Because the optimality of a learning rate depends on how much informative on both perspectives and how to exploit them, we derive three research objectives: (i) how to describe {the importance of }the task with a new user-item interaction, (ii) how to identify the role of each parameter for the task, then (iii) how to determine the optimal learning rate for each pair of interactions and parameters based on their mutual relevance. 

Corresponding to the three research questions, \algname{} goes through the following three steps, as shown in Figure \ref{fig:MeLON_idea_procedure}.
First, because exploiting the connections from the new user-item interaction is very helpful to mitigate the \emph{data sparsity} issue, \algname{} employs a graph attention network\,(GAT)\,\cite{velivckovic2018graph} to represent {the importance of} each new user-item interaction along with previous user-item interactions. 
Then, an explicit neural mapper dynamically captures the role of a parameter, assessing its contribution to the new user-item interaction by the loss and gradient.
Last, the two representations---for an interaction and a parameter---are jointly considered to generate the optimal learning rate specific to the interaction-parameter pair.
{\algname{} repeats the three steps for every online update, following the \emph{learning-to-learn} philosophy of meta-learning. That is, the meta-model \algname{} learns to provide the learning rates such that the recommender model updated using those learning rates quickly grasps what users want now and succeeds recommendations for them in the future.}

The effectiveness of \algname{} is extensively evaluated on two famous recommender algorithms using three real-world online service datasets in a comparison with six update strategies.
In short, the results show that \algname{} successfully improves the recommendation accuracy by up to {29.9\%} in term of HR@5.
Such capability of \algname{} is empowered by two-directional flexibility under learning-to-learn strategy, which is further supported by the theoretical analysis and ablation study.

\section{Preliminary and Related Work}
\label{sec:preliminary}

Online recommenders build a pre-trained model using previous user-item interactions, and the pre-trained model is continuously updated in response to incoming user-item interactions. A deep neural network\,(DNN) is widely used for a recommender, and it is updated for each \emph{mini-batch}\,\cite{ruder2016overview} to quickly adapt to users' up-to-date interest. Let $x = (t, u, i)$ denote a user-item interaction between a user $u$ and an item $i$ at time $t$. Suppose that a mini-batch $\mathcal{B}_t$ consists of $n$ new user-item interactions at time $t$.
Then, the recommender at time $t$, parameterized by $\Theta_{t} = \{\theta_{t,1}, \ldots,\theta_{t,M}\}$ where $M$ is the total number of parameters, is updated by
\begin{equation}
\begin{split}
{\Theta}_{t+1} &= {\Theta}_t - {\eta}\nabla_{{\Theta}_t}\sum_{x \in \mathcal{B}_t}\frac{1}{n}\mathcal{L}_{\Theta_{t}}(x)\\
&= {\Theta}_t - \nabla_{{\Theta}_t}\mathcal{L}_{\Theta_t}(\mathcal{B}_t)^{\top}\boldsymbol{W}.
\label{eq:Loss_default}
\end{split}
\end{equation}
Here, ${\eta}$ is a learning rate, and 
${\mathcal{L}_{\Theta_t}(\mathcal{B}_t)}\in \mathbb{R}^{n}$ denotes the loss of new user-item interactions in the mini-batch $\mathcal{B}_t$ by the recommender model $\Theta_t$ under any objective function $\mathcal{L}$ such as mean squared error\,(MSE) or Bayesian personalized ranking\,(BPR)\,\cite{rendle2012bpr}. 
The \emph{learning rate matrix} $\boldsymbol{W} \in \mathbb{R}^{n \times M}$ is used to represent the learning rate for a parameter $\theta_{m}$ in response to each user-item interaction $x$, where all the learning rates (i.e., all elements of $\boldsymbol{W}$) are typically set equally to $w = \frac{\eta}{n}$. 
Then, the overall performance is derived by evaluating each recommender snapshot for a given mini-batch at each time step,
\begin{align}
\min_{\{\Theta_t\}_{t=1}^{T}} \sum^{T}_{t=1}\sum_{x \in \mathcal{B}_{t}} \mathcal{L}_{\Theta_t}(x) = \min_{\{\Theta_t\}_{t=1}^{T}} \sum^{T}_{t=1}\mathcal{L}_{\Theta_t}(\mathcal{B}_t).
\label{eq:Loss_total}
\end{align}

The two directions---importance reweighting and meta-optimization---for online recommender updates are characterized by the construction of the learning rate matrix $\boldsymbol{W}$.

\subsection{Importance Reweighting}

Instead of assigning the equal importance $1 / n$ to each user-item interaction as in Eq.~\eqref{eq:Loss_default}, \emph{importance reweighting}\,\cite{he2016fast, shu2019meta} assigns a different importance determined by a reweighting function $\phi^{I}(\cdot)$, 
\begin{align}
\begin{split}
{\Theta}_{t+1} &= {\Theta}_t - {\eta}\nabla_{{\Theta}_t}\sum_{x \in \mathcal{B}_t} \mathcal{L}_{\Theta_t}(x)\cdot \underbrace{\phi^{I}\big(\mathcal{L}_{\Theta_t}(x)\big)}_\text{interaction-wise}\\
&= {\Theta}_t - \nabla_{{\Theta}_t}\mathcal{L}_{\Theta_t}(\mathcal{B}_t)^{\top}\boldsymbol{W}^{I},
\label{eq:Loss_importance_reweighting}
\end{split}
\end{align}
where $\phi^{I}(\cdot)$ receives the loss of each user-item interaction as its input.
That is, $\boldsymbol{W}^{I} \in \mathbb{R}^{n \times M}$ is constructed such that each row has the same value returned by $\phi^{I}(\cdot)$. 
The representative methods differ in the detail of $\phi^{I}(\cdot)$, as follows:
\begin{itemize}[leftmargin=9pt, noitemsep]
    \item eALS\,\cite{he2016fast} applies a heuristic rule that assigns a weight for each new user-item interaction. Typically, a high weight is set to learn the current user interest.
    \item MWNet\,\cite{shu2019meta} maintains an external meta-model that adaptively assesses the importance of a given user-item interaction for model update that lets the updated model minimize the loss on meta-data\,(e.g., next recommendation in online update).
\end{itemize}

However, this scheme does not support the varying role of a parameter for different tasks.


\subsection{Meta-Optimization}
On the other hand, \emph{meta-optimization}\,\cite{ravi2016optimization, li2017meta, du2019sequential, zhang2020retrain} aims at adjusting the learning rate of each recommender parameter $\theta_{t,m}$ via a
learning rate function $\phi^{P}(\cdot)$, 
\begin{equation}
\begin{split}
{\Theta}_{t+1} &={\Theta}_{t} - \underbrace{\phi^{P}\Big(\mathcal{L}_{\Theta_t}(\mathcal{B}_t), {\Theta}_{t}\Big)}_{\text{parameter-wise}} \cdot\nabla_{{\Theta}_{t}}\sum_{x \in \mathcal{B}_t}\frac{1}{n}\mathcal{L}_{\Theta_t}(x)\\
&= {\Theta}_t - \nabla_{{\Theta}_t}\mathcal{L}_{\Theta_t}(\mathcal{B}_t)^{\top}\boldsymbol{W}^{P},
\label{eq:Loss_parameter_optimizer}
\end{split}
\end{equation}
where the function  $\phi^{P}(\cdot)$ receives the training loss of a mini-batch and the recommender parameters $\Theta_{t}$ as its input. That is, $\boldsymbol{W}^{P} \in \mathbb{R}^{n \times M}$ is constructed such that each column has the same value returned by $\phi^{P}(\cdot)$. 
Again, the representative algorithms differ in the detail of $\phi^{P}(\cdot)$, as follows:
\begin{itemize}[leftmargin=9pt, noitemsep]
    \item S\textsuperscript{2}Meta\,\cite{du2019sequential} exploits MetaLSTM\,\cite{ravi2016optimization} to decide how much to forget a parameter's previous knowledge and to learn new user-item interactions via the gating mechanism of LSTM\,\cite{hochreiter1997long}.
    \item MetaSGD\,\cite{li2017meta} maintains one learnable parameter for each model parameter to adjust its learning rate based on the loss.
    \item SML\,\cite{zhang2020retrain} maintains a convolutional neural network\,(CNN)-based meta-model with pretrained and fine-tuned parameters. It decides how much to combine the knowledge for previous interactions and that for new user-item interactions for each parameter.
\end{itemize}

Contrary to importance reweighting, this scheme does not support the varying importance of user-item interactions.


\subsection{Difference from Previous Work}

While previous update strategies achieve only \emph{one}-directional flexibility, i.e., $\phi^{1D} \in \{\phi^{I}, \phi^{P}\}$, we aim at developing an online update strategy ${\phi}^{2D}$ that provides \emph{two}-directional flexibility for the learning rates to be adaptive in \emph{both} data and parameter perspectives,
\begin{equation}
\begin{split}
{\Theta}_{t+1} &={\Theta}_{t} - \nabla_{{\Theta}_{t}}\!\!\sum_{x \in \mathcal{B}_t}\!\!\mathcal{L}_{\Theta_t}(x)\cdot \underbrace{{\phi}^{2D}\big(x,\mathcal{L}_{\Theta_t}(x),\Theta_t\big)}_\text{interaction-/parameter-wise}\\
&= {\Theta}_t - \nabla_{{\Theta}_t}\mathcal{L}_{\Theta_t}(\mathcal{B}_t)^{\top}{\boldsymbol{W}^{2D}},
\label{eq:Loss_MeLON}
\end{split}
\end{equation}
where the function receives an individual user-item interaction, the training loss of the interaction, and the recommender parameters $\Theta_{t}$ as its input, which are essential ingredients to be adaptive to both user-item interactions and parameters. That is, $\boldsymbol{W}^{2D}\in \mathbb{R}^{n \times M}$ is constructed such that each entry can be filled with a different value returned by $\phi^{2D}(\cdot)$ even when either a user-item interaction or parameter is identical to other entries as in Figure \ref{fig:motivating_example}(c),
\begin{align}
\nonumber
& (x = x^{\prime}) \wedge (\theta_m\!\!\neq\!\theta_{m^{\prime}}) \\ 
& \!\nRightarrow\! {\phi}^{2D}\big(x,\mathcal{L}_{\Theta_t}(x),\theta_m\big) = {\phi}^{2D}\big(x^{\prime},\mathcal{L}_{\Theta_t}(x^{\prime}),\theta_{m^\prime}\big), \label{eq:flex:one} \\
\nonumber 
& (x \neq x^{\prime}) \wedge (\theta_m\!\!=\!\theta_{m^{\prime}})\\
& \!\nRightarrow\! {\phi}^{2D}\big(x,\mathcal{L}_{\Theta_t}(x),\theta_m\big)={\phi}^{2D}\big(x^{\prime},\mathcal{L}_{\Theta_t}(x^{\prime}),\theta_{m^\prime}\big). \label{eq:flex:two}
\end{align}



\section{Methodology: \algname{}}
\label{sec:methodology}

\algname{} is a \emph{meta-model} that determines the optimal learning rate for each recommender parameter regarding a user-item interaction. Figure \ref{fig:recommender_update} shows the collaboration between a recommender model and the meta-model \algname{}. For each iteration of online update, a recommender model provides its parameters $\Theta_t$ and the loss $\mathcal{L}_{\Theta_t}({\mathcal{B}_t})$ of the current batch $\mathcal{B}_t$ to \algname{}; then, additionally using the previous interaction history, \algname{} provides the learning rate matrix $\boldsymbol{W}^{2D}$, which is learned to reduce $\mathcal{L}_{\Theta_t}({\mathcal{B}_t})$ as much as possible, to the recommender model; finally, the recommender model is updated using $\boldsymbol{W}^{2D}$ for $\mathcal{B}_t$. Please refer to the Section A of the supplementary material for more details. The internal procedure of \algname{} is described according to the three research questions: (i) representing the relevance between a user and an item for each user-item interaction, (ii) representing the role of each parameter for a user-item interaction, and (iii) determining the learning rate for each pair of user-item interactions and parameters. 

\begin{figure}[t!]
\begin{center}
\includegraphics[width=0.47\textwidth]{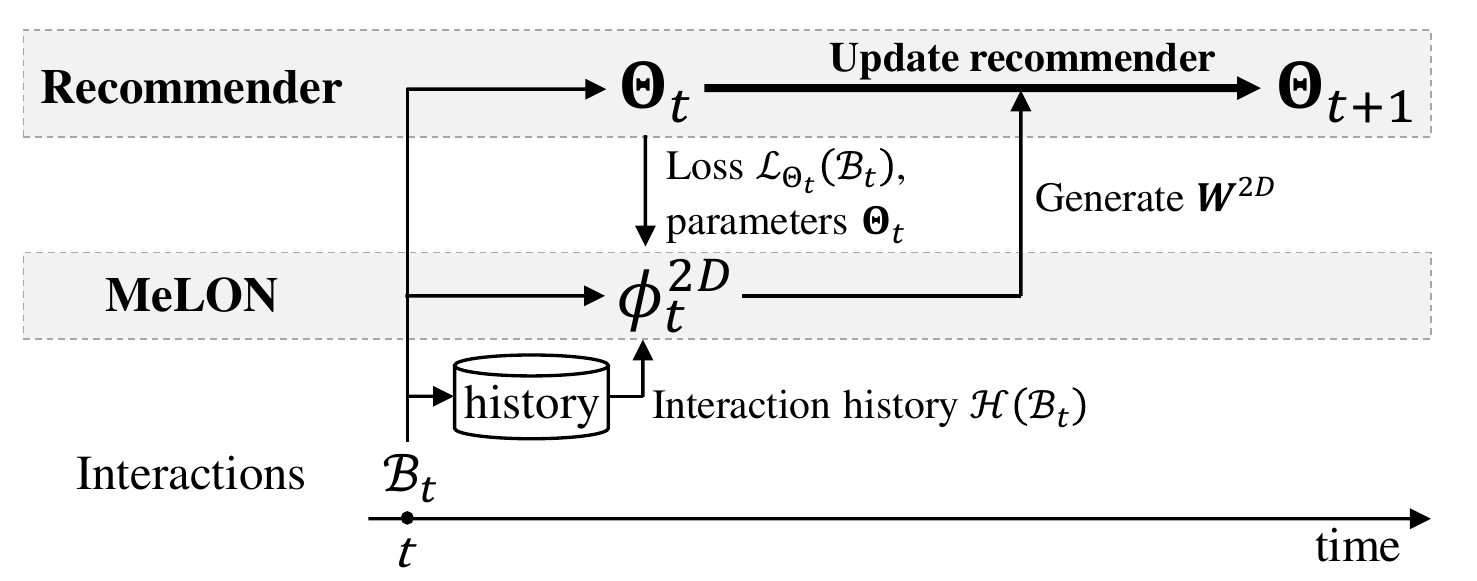}
\end{center}
\vspace*{-0.3cm}
\caption{Online update procedure with the meta-model \algname{} involved.}
\label{fig:recommender_update}
\vspace*{-0.3cm}
\end{figure}

\subsection{Step I: Representing User-Item Interaction}

Because a single user-item interaction may not contain sufficient information, we utilize the information from the previous interaction history by adding the users and items connected to the user-item interaction. More specifically, the latent representation of the user-item interaction is derived using a graph attention network\,(GAT) on the bipartite graph that represents the interactions between users and items received until the current time. The bipartite graph for a user-item interaction $x$ is constructed from the users and items in the \emph{interaction history} in Definition \ref{def:history}.

\begin{definition}\,{\sc (Interaction History)}
Given a user-item interaction $x = (t, u, i)$, the \emph{user interaction history of $u$} is the set of items interacted with $u$ before $t$, $\mathcal{H}_{user}(x) = \{ i^\prime \mid \exists (t^\prime, u, i^\prime) \in \mathcal{X} \text{ s.t. } t^\prime < t \}$, where $\mathcal{X}$ is the entire set of user-item interactions; similarly, the \emph{item interaction history of $i$} is the set of users interacted with $i$ before $t$, $\mathcal{H}_{item}(x) = \{ u^\prime \mid \exists(t^\prime, u^\prime, i) \in \mathcal{X} \text{ s.t. } t^\prime < t \}$. \hfill $\Box$
\label{def:history}
\end{definition}

For the bipartite graph, the users in $\mathcal{H}_{user}(x)$ constitute the user side, and the items in $\mathcal{H}_{item}(x)$ constitute the item side. Here, each user (or item) node is represented by the user (or item) embedding used in the recommender model. An edge is created for each of the previous user-item interactions, and its weight is determined by the attention score between them. Then, a user (or item) embedding is \emph{extended} using the connections to the other side on the bipartite graph, as specified in Definition \ref{def:extended_embedding}.

\begin{definition}\,{\sc (Extended Embedding)}
\label{def:extended_embedding}
Given a user-item interaction $x = (t, u, i)$, let $\mathbf{e}_{u}$ and $\mathbf{e}_{i^\prime}$ be the embeddings of $u$ and $i^\prime \in \mathcal{H}_{user}(x)$. Then, the \emph{extended embedding} of $u$, $\tilde{\mathbf{e}}_{u}$, is defined as
\begin{equation}
\tilde{\mathbf{e}}_{u} = {\rm ReLU}(W_{user}\cdot [\mathbf{e}_{u}, \sum_{i^{\prime} \in \mathcal{H}_{user}(x)}\!\!\!\!\!\!\!\!\!\alpha_{ui^{'}}\mathbf{e}_{i^\prime}] + \mathbf{b}_{user}),
\end{equation}
where $W_{user}$ and $\mathbf{b}_{user}$ are a learnable weight matrix and a bias vector.
Here, $\alpha_{ui^\prime}$ indicates the attention score for $i^\prime$ and is derived by the GAT, as follows:
\begin{equation}
{\alpha}_{ui^{'}} = {\rm softmax}\big({\rm LeakyReLU}\big([\mathbf{e}_u,\mathbf{e}_{i^{'}}]^{\top}\mathbf{a}_U \big)\big),
\label{eq:edge_ui}
\end{equation}
where $\mathbf{a}_U$ is a learnable attention vector. In addition, the \emph{extended embedding} of an item $i$, $\tilde{\mathbf{e}}_{i}$, is defined in the same way to the opposite direction.
\hfill $\Box$
\end{definition}

Last, the two extended embeddings, $\tilde{\mathbf{e}}_{u}$ and $\tilde{\mathbf{e}}_{i}$, are concatenated and gone through a linear mapping to learn the relevance between the user and the item, as specified in Definition \ref{def:interaction_representation}. As a result, the \emph{interaction representation} contains not only the rich information about a user and an item but also the relevance between them.

\begin{definition}\,{\sc (Interaction Representation)}
\label{def:interaction_representation}
Given a user-item interaction $x = (t, u, i)$, let $\tilde{\mathbf{e}}_{u}$ and $\tilde{\mathbf{e}}_{i}$ be the extended embeddings of $u$ and $i$, respectively. The \emph{interaction representation} of $x$, $\mathbf{h}_{x}$, is defined by
\begin{equation}
\mathbf{h}_{x} = {\rm ReLU}\big(W_x \cdot [\tilde{\mathbf{e}}_{u},\tilde{\mathbf{e}}_i] + \mathbf{b}_{x} \big),
\label{eq:propagation_item_to_user}
\end{equation}
where $W_x$ and $\mathbf{b}_{x}$ are a learnable weight matrix and a bias vector. \hfill $\Box$
\end{definition}

\subsection{Step II: Representing Parameter Role}

Because it is well known that a parameter in a neural network has different relevance toward different tasks\,(user-item interactions in our study)\,\cite{bengio2013representation}, we contend that a parameter has a different role for each user-item interaction. The \emph{role} of a parameter can be roughly defined as its degree of impact on users (or items) of common characteristics. For example, a specific parameter may have a high impact on action films, while another parameter may have a high impact on romance films. 

To help find a parameter role, the latent representation of a parameter is derived using three types of information: the current value of a parameter $\theta_{t,m}$, the loss $\mathcal{L}_{\Theta_t}(x)$ of a recommender model for a given user-item interaction $x$, and the gradient $\nabla_{\theta_{t,m}}\mathcal{L}_{\Theta_t}(x)$ of the loss with respect to the parameter.  The loss represents how much the recommender model parameterized by $\Theta_t$ has not learned that user-item interaction. Each gradient represents how much the corresponding parameter needs to react to that loss; we expect that relevant parameters usually have higher gradients than irrelevant ones. Thus, putting them together, they can serve as useful information for determining a parameter role.

Symmetric to the interaction representation in Definition \ref{def:interaction_representation}, the \emph{role representation} is obtained through a multi-layer perceptron\,(MLP), as specified in Definition \ref{def:role_representation}. Because the magnitude of the loss and gradient varies across the pair of interactions and parameters, we apply a preprocessing technique\,\cite{ravi2016optimization, andrychowicz2016learning} to adjust the scale of the loss and gradient as well as to separate their magnitude and sign. As a result, the output of the MLP, $\mathbf{h}_{\theta_{t,m}}$, is regarded to represent the role of $\theta_{t,m}$ with respect to the given user-item interaction $x$.

\begin{definition}\,{\sc (Role Representation)}
\label{def:role_representation}
Given a parameter $\theta_{t,m}$ and a user-item interaction $x$, the \emph{role representation} of $\theta_{t,m}$, $\mathbf{h}_{\theta_{t,m}}$ is defined by
\begin{equation}
\mathbf{h}_{\theta_{t,m}} = {\rm MLP}(\big[\theta_{t,m},~\mathcal{L}^{}_{\Theta_t}(x),~ \nabla_{\theta_{t,m}}\mathcal{L}^{}_{\Theta_t}(x)\big]),
\label{eq:neural_mapper}
\end{equation}
where the MLP consists of $L$ linear mapping layers each followed by the ReLU activation function. \hfill $\Box$
\end{definition}

\subsection{Step III: Adapting Learning Rate}

The resulting two representations, $\mathbf{h}_{x}$ and $\mathbf{h}_{\theta_{t,m}}$, respectively, contain rich information about the importance of the user-item interaction $x$ and the role of the parameter $\theta_{t,m}$. Hence, we employ a linear mapping layer which fuses the two representations and adapts the learning rate $w_{x,\theta_{t,m}}^{lr}$ to the given interaction-parameter pair, as follows:
\begin{equation}
w_{x,\theta_{t,m}}^{lr} = \sigma(W_{lr}\cdot [\mathbf{h}_{x}, \mathbf{h}_{\theta_{t,m}}] + \mathbf{b}_{lr}),
\end{equation}
where $\sigma$ is a sigmoid function. The learning rate is likely to be high if the user-item interaction is important while the parameter plays a key role to the interaction, so that the parameter is quickly adapted to the interaction.
Then, the learning rate is used to update the current parameter $\theta_{t,m}$ for the user-item interaction $x$, as follows:
\begin{equation}
{\theta}_{t+1,m} = \theta_{t,m} - w_{x,\theta_{t,m}}^{lr} \cdot \nabla_{\theta_{t,m}}\mathcal{L}_{\Theta_t}(x).
\label{eq:MeLON_update_equation}
\end{equation}
\section{Theoretical Analysis on Update Flexibility}
\label{sec:theory}

Our suggested online update strategy \algname{}, ${\phi}^{2D}$, leaves a question of how much benefit it can bring compared with the previous two strategies $\phi^{I}$ and $\phi^{P}$.
As an effort to resolve it, we present a theoretical analysis of the advantage of flexible update in terms of rank, where the rank of a learning rate matrix $rank(\boldsymbol{W})$ demonstrates how flexible an update can be via $\boldsymbol{W}$.
That is, when its rank can be higher, $\boldsymbol{W}$ can support more flexible updates of parameters in response to new interactions.
The previous strategies limit the rank to $1$ since, as discussed above, they provide an identical learning rate either to every interaction or to every parameter. In \algname{}, the rank can be higher since the learning rates are adapted to each interaction-parameter pair.
In this regard, we show that previous strategies may suffer from large optimality gap with an optimal learning rate matrix $\boldsymbol{W}^{*}$, while  the gap can be reduced by \algname{}.

We denote the recommender parameters updated by $\boldsymbol{W}$ as $\hat{\Theta}$  and the optimal parameters as $\Theta^{*}$.
Then, the optimality gap between the two sets of parameters $\|\Theta^{*} - \hat{\Theta}\|_2$ is dependent on the gap between the learning rate matrices $\|\boldsymbol{W}^{*} - \boldsymbol{W}\|_2$ in terms of spectral norm as follows:
\begin{equation}
\nonumber
\begin{split}
\|\Theta^{*}\!\!- \hat{\Theta}\|_2 &= \|(\Theta - \nabla_{\Theta}\mathcal{L}\cdot\boldsymbol{W^{*}})- (\Theta - \nabla_{\Theta}\mathcal{L}\cdot\boldsymbol{W})\|_2\\
&= \|-(\nabla_{\Theta}\mathcal{L}\cdot\boldsymbol{W^{*}})-(-\nabla_{\Theta}\mathcal{L}\cdot\boldsymbol{W})\|_2\\ 
&= \|\nabla_{\Theta}\mathcal{L}\cdot(\boldsymbol{W^{*}}-\boldsymbol{W})\|_2\\
&\leq \|\nabla_{\Theta}\mathcal{L}\|_2\cdot\|(\boldsymbol{W^{*}}-\boldsymbol{W})\|_2.
\end{split}
\end{equation}
Then, a lower bound of $\|\boldsymbol{W}^{*} - \boldsymbol{W}\|_2$ is obtained from the singular values $\sigma$ of $\boldsymbol{W}^{*}$, as formalized in Lemma \ref{lemma:Eckart–Young–Mirsky}.
\begin{lemma}
\emph{\cite{eckart1936approximation}} Given $\boldsymbol{W}^{*}$ with its singular value decomposition $\boldsymbol{U\Sigma V}$ and $k\in \{1,\cdots, rank(\boldsymbol{W}^{*})-1\}$, let 
$\boldsymbol{W}^{*}_k = \sum_{r=1}^{k}\sigma_{r}\boldsymbol{U}_{r}\boldsymbol{V}_{r}$, where $\sigma_r$ is the $r$-th largest singular value. 
Then, $\boldsymbol{W}^{*}_k$ is the best rank-$k$ approximation of $\boldsymbol{W}^{*}$ in terms of spectral norm, 
\label{lemma:Eckart–Young–Mirsky}
\begin{align}
\nonumber
\min_{\boldsymbol{W}:rank(\boldsymbol{W})=k}\!\!\|\boldsymbol{W}^{*}\!\!-\!\boldsymbol{W}\|_2 = \|\boldsymbol{W}^{*}\!\!-\!\boldsymbol{W}^{*}_k\|_2 = \sigma_{k+1}.
\end{align}
\begin{proof}
See \cite{eckart1936approximation}.
\end{proof}
\end{lemma}
Based on Lemma \ref{lemma:Eckart–Young–Mirsky}, we show that a flexible update strategy $\phi^{2D}$ can enjoy smaller optimality gap than the previous strategies, which are one-directionally flexible. 
\begin{lemma}
For $\boldsymbol{W}^{I}$ and $\boldsymbol{W}^{P}$ (see Eq. \eqref{eq:Loss_importance_reweighting} and Eq. \eqref{eq:Loss_parameter_optimizer}), $rank(\boldsymbol{W}^{I})=rank(\boldsymbol{W}^{P})=1$ holds. 
\end{lemma}
\vspace*{-0.25cm}
\begin{proof}
Every column of $\boldsymbol{W}^{I}$ equals to $\phi^{I}\big(\mathcal{L}_{\Theta_t}(x)\big)\in \mathcal{R}^{n}$, and
every row of $\boldsymbol{W}^{P}$ equals to $\phi^{P}\big(\mathcal{L}_{\Theta_t}(\mathcal{B}_t), {\Theta}_{t}\big)\in \mathcal{R}^{M}$.
Hence, $rank(\boldsymbol{W}^{I})=rank(\boldsymbol{W}^{P})=1$ holds.
\end{proof}
\vspace*{-0.25cm}
\begin{theorem}
For $\boldsymbol{W}^{1D}\in \{\boldsymbol{W}^{I},\boldsymbol{W}^{P}\}$ (see Eq. \eqref{eq:Loss_importance_reweighting} and Eq. \eqref{eq:Loss_parameter_optimizer}) and $\boldsymbol{W}^{2D}$ (see Eq. \eqref{eq:Loss_MeLON}), the following inequality holds: 
\begin{equation}
\nonumber
\min_{\boldsymbol{W}^{1D}}\|\boldsymbol{W}^{*} - \boldsymbol{W}^{1D}\|_2 \geq \min_{\boldsymbol{W}^{2D}}\|\boldsymbol{W}^{*} - \boldsymbol{W}^{2D}\|_2.
\vspace*{-0.5cm}
\end{equation}
\label{lemma:rank_1_previous_strategies}
\end{theorem}
\begin{proof}
Lemma~\ref{lemma:Eckart–Young–Mirsky}, Lemma \ref{lemma:rank_1_previous_strategies}, and $\boldsymbol{W}^{1D}\in \{\boldsymbol{W}^{I},\boldsymbol{W}^{P}\}$ imply
$\min_{\boldsymbol{W}^{1D}}\|\boldsymbol{W}^{*} - \boldsymbol{W}^{1D}\| = \sigma_{2}$.
On the other hand, $rank(\boldsymbol{W}^{2D}) \geq 1$,\footnote{Specifically, by Eq.~\eqref{eq:flex:one}  and Eq.~\eqref{eq:flex:two}, $rank(\boldsymbol{W}^{2D})$ is not necessarily one and can be greater than one.}
Thus, by Lemma~\ref{lemma:Eckart–Young–Mirsky}, $\min_{\boldsymbol{W}^{2D}}\|\boldsymbol{W}^{*} - \boldsymbol{W}^{2D}\| \leq \sigma_{2}$, which concludes the proof, holds.
\end{proof}
In the experiments, we empirically validate the advantage of the \emph{two}-directional flexibility of $\boldsymbol{W}^{2D}$. 

\section{Evaluation}
\label{sec:evaluation}

\newcolumntype{L}[1]{>{\raggedright\let\newline\\\arraybackslash\hspace{0pt}}m{#1}}
\newcolumntype{X}[1]{>{\centering\let\newline\\\arraybackslash\hspace{0pt}}p{#1}}
\newcolumntype{Y}[1]{>{\let\newline\\\arraybackslash\hspace{1pt}}m{#1}}
\newcolumntype{C}[1]{>{\centering\arraybackslash}p{#1}}
\newcolumntype{M}[1]{>{\centering\arraybackslash}m{#1}}

Our evaluation was conducted to support the following:
\vspace*{-0.05cm}
\begin{itemize}[leftmargin=12pt] 
\item The performance improvement by \algname{} is {consistent} for various datasets and recommenders.
\item \algname{} helps recommenders quickly adapt to users' up-to-date interest over time.
\item The two-directional flexibility in \algname{} is very effective for recommendation. 
\item The training overhead of \algname{} is affordable.
\end{itemize}

\subsection{Experiment Settings}

\begin{table}[t!]
\center
\small
\begin{tabular}{crrr}
\toprule
 Dataset           &  Users      &  Items   &  Interactions \\ \midrule
 Adressa           & 29,589     & 1,457    & 1,191,114  \\  
 Amazon            & 91,013     & 118,031  & 3,625,349  \\ 
 Yelp              & 60,543     & 74,249   & 2,880,520  \\  
\bottomrule
\end{tabular}
\vspace*{-0.15cm}
\caption{Summary of the three real-world datasets.} 
\vspace*{-0.6cm}
\label{Tab:dataset_summary}
\end{table}

\subsubsection{Datasets.}
We used three real-world online recommendation benchmark datasets: Adressa\,\cite{gulla2017adressa}, Amazon\,\cite{ni2019justifying}, and Yelp\footnote{\url{https://www.kaggle.com/yelp-dataset/yelp-dataset}}, as summarized in Table \ref{Tab:dataset_summary}. The duration that a user's interest persists varies across datasets; relatively short duration for news in Adressa, typically longer duration for locations in Yelp, and in-between them for products in Amazon. 

\subsubsection{Algorithms and Implementation Details.}
For the base recommender, we used two popular personalized recommender algorithms: BPR\,\cite{koren2009matrix, rendle2012bpr} and NCF\,\cite{he2017neural}. 
For the online training strategy, we compared \algname{} with six update methods, namely Default, eALS\,\cite{he2016fast}, MWNet\,\cite{shu2019meta}, MetaSGD\,\cite{li2017meta}, S$^2$Meta\,\cite{du2019sequential}, and SML\,\cite{zhang2020retrain}.
``Default'' is the standard fine-tuning strategy, and the remaining methods are based on either importance reweighting or meta-optimization. Hence, $14$ combinations of two recommenders and seven update strategies were considered for evaluation. The experiment setting was exactly the same for all the combinations.


\begin{table*}[htb!]
\resizebox{\textwidth}{!}{
\begin{tabular}{c|l|ccc|ccc|ccc|ccc}\toprule
\centering{\multirow{2}{*}{Dataset}} & \centering{\multirow{2}{*}{\, Method}} & \multicolumn{6}{c}{NCF}  \vline & \multicolumn{6}{c}{BPR}\\ \cline{3-14}
\rule{0pt}{2.5ex}& & HR@5 &  HR@10 & HR@20 & \!\!NDCG@5\!\! & \!\!NDCG@10\!\! & \!\!\!NDCG@20\!\!\! & HR@5 &  HR@10 & HR@20 & \!\!NDCG@5\!\! & \!\!NDCG@10\!\! & \!\!\!NDCG@20\!\!\! 
\\[-0.25ex] \midrule

\multirow{7}{*}{\centering\rotatebox[origin=c]{90}{\makecell[c]{Adressa}}}
& Default
& 0.334$\pm$0.007  & 0.407$\pm$0.009  & 0.502$\pm$0.017 & 0.283$\pm$0.005  & 0.306$\pm$0.006 & 0.330$\pm$0.007
& 0.292$\pm$0.002  & 0.359$\pm$0.002  & 0.422$\pm$0.001 & 0.250$\pm$0.001  & 0.272$\pm$0.001 & 0.288$\pm$0.001
\\

& eALS
& \underline{0.664$\pm$0.006}  & \underline{0.750$\pm$0.006}  & \underline{0.826$\pm$0.004} & \underline{0.542$\pm$0.004}  & \underline{0.570$\pm$0.004} & \underline{0.589$\pm$0.003}
& \underline{0.443$\pm$0.008}  & \underline{0.520$\pm$0.011}  & 0.613$\pm$0.016 & \underline{0.371$\pm$0.006} & \underline{0.396$\pm$0.007} & \underline{0.419$\pm$0.008}
\\

& MWNet
& 0.325$\pm$0.009  & 0.392$\pm$0.010 & 0.480$\pm$0.014 & 0.276$\pm$0.007  & 0.297$\pm$0.007 & 0.319$\pm$0.008
& 0.289$\pm$0.001  & 0.356$\pm$0.001 & 0.421$\pm$0.008 & 0.248$\pm$0.001  & 0.269$\pm$0.001 & 0.285$\pm$0.000
\\

& MetaSGD  
& 0.275$\pm$0.000  & 0.406$\pm$0.002  & 0.686$\pm$0.002 & 0.229$\pm$0.000  & 0.270$\pm$0.000 & 0.340$\pm$0.000
& 0.276$\pm$0.002  & 0.405$\pm$0.002  & \underline{0.686$\pm$0.002} & 0.230$\pm$0.001  & 0.271$\pm$0.001 & 0.340$\pm$0.001
\\

& S$^2$Meta
& 0.276$\pm$0.016  & 0.396$\pm$0.039  & 0.548$\pm$0.062 & 0.221$\pm$0.008  & 0.260$\pm$0.015 & 0.298$\pm$0.021
& 0.278$\pm$0.015  & 0.401$\pm$0.037  & 0.559$\pm$0.060 & 0.223$\pm$0.007  & 0.262$\pm$0.014 & 0.302$\pm$0.020
\\

& SML
& N/A  & N/A  & N/A  & N/A   & N/A   & N/A    
& 0.270$\pm$0.001  & 0.330$\pm$0.002  & 0.399$\pm$0.001 & 0.235$\pm$0.001  & 0.255$\pm$0.001 & 0.272$\pm$0.001
\\[0.5ex]\cline{2-14}\rule{0pt}{2ex}

& \algname{}
& \textbf{0.863$\pm$0.004}  & \textbf{0.954$\pm$0.000} & \textbf{0.982$\pm$0.000} & \textbf{0.626$\pm$0.011}  & \textbf{0.656$\pm$0.010} & \textbf{0.664$\pm$0.010}
& \textbf{0.877$\pm$0.004}  & \textbf{0.958$\pm$0.001} & \textbf{0.983$\pm$0.000} & \textbf{0.671$\pm$0.007}  & \textbf{0.698$\pm$0.006} & \textbf{0.705$\pm$0.005}
\\
\midrule

\multirow{7}{*}{\rotatebox[origin=c]{90}{\makecell[c]{Amazon}}}
& Default
& 0.168$\pm$0.001 & 0.244$\pm$0.002 & 0.359$\pm$0.006 & 0.115$\pm$0.001 & 0.140$\pm$0.001 & 0.168$\pm$0.001
& 0.246$\pm$0.002 & 0.339$\pm$0.003 & 0.457$\pm$0.003 & 0.172$\pm$0.002 & 0.202$\pm$0.002 & 0.231$\pm$0.001
\\

& eALS
& 0.219$\pm$0.009 & 0.323$\pm$0.014 & 0.462$\pm$0.019 & 0.148$\pm$0.006 & 0.182$\pm$0.007 & 0.216$\pm$0.008
& \underline{0.327$\pm$0.002} & \underline{0.425$\pm$0.002} & \underline{0.542$\pm$0.001} & \underline{0.238$\pm$0.002} & \underline{0.270$\pm$0.002} & \underline{0.299$\pm$0.001}
\\

& MWNet
& 0.169$\pm$0.002 & 0.247$\pm$0.005 & 0.366$\pm$0.010 & 0.116$\pm$0.001 & 0.142$\pm$0.002 & 0.171$\pm$0.003
& 0.244$\pm$0.001 & 0.339$\pm$0.002 & 0.456$\pm$0.003 & 0.171$\pm$0.001 & 0.201$\pm$0.001 & 0.231$\pm$0.001
\\

& MetaSGD
& 0.151$\pm$0.003  & 0.214$\pm$0.005 & 0.317$\pm$0.006 & 0.104$\pm$0.002  & 0.125$\pm$0.003 & 0.150$\pm$0.002
& 0.148$\pm$0.002  & 0.214$\pm$0.005 & 0.314$\pm$0.005 & 0.103$\pm$0.001  & 0.123$\pm$0.002 & 0.148$\pm$0.001
\\

& S$^2$Meta
& \underline{0.292$\pm$0.005}  & \underline{0.390$\pm$0.007} & \underline{0.497$\pm$0.010} & \underline{0.192$\pm$0.003}  & \underline{0.224$\pm$0.004} & \underline{0.250$\pm$0.004}
& 0.270$\pm$0.029  & 0.367$\pm$0.037 & 0.474$\pm$0.040 & 0.178$\pm$0.001  & 0.209$\pm$0.021 & 0.237$\pm$0.216
\\

& SML
& N/A  & N/A  & N/A  & N/A   & N/A   & N/A   
& 0.220$\pm$0.001  & 0.307$\pm$0.001 & 0.409$\pm$0.001 & 0.153$\pm$0.001  & 0.181$\pm$0.001 & 0.207$\pm$0.000
\\[0.5ex]\cline{2-14}\rule{0pt}{2ex}

& MeLON
& \textbf{0.324$\pm$0.040}  & \textbf{0.519$\pm$0.034}  & \textbf{0.807$\pm$0.014} & \textbf{0.225$\pm$0.031}  & \textbf{0.287$\pm$0.028} & \textbf{0.360$\pm$0.020}
& \textbf{0.363$\pm$0.016}  & \textbf{0.506$\pm$0.013}  & \textbf{0.650$\pm$0.007} & \textbf{0.248$\pm$0.012}  & \textbf{0.294$\pm$0.011} & \textbf{0.330$\pm$0.009}
\\
\midrule

\multirow{7}{*}{\rotatebox[origin=c]{90}{\makecell[c]{Yelp}}}
& Default
& \underline{0.659$\pm$0.001}  & 0.816$\pm$0.002  & 0.923$\pm$0.002 & \underline{0.477$\pm$0.001}  & \underline{0.528$\pm$0.001} & \underline{0.555$\pm$0.002}
& 0.600$\pm$0.003  & 0.766$\pm$0.004 & 0.883$\pm$0.005 & 0.426$\pm$0.002  & 0.480$\pm$0.002 & 0.509$\pm$0.002
\\

& eALS
& 0.618$\pm$0.009  & 0.781$\pm$0.014  & 0.901$\pm$0.001 & 0.438$\pm$0.006  & 0.491$\pm$0.007 & 0.521$\pm$0.001
& \textbf{0.677$\pm$0.003}  & \textbf{0.831$\pm$0.003} & \textbf{0.922$\pm$0.002} &\textbf{0.488$\pm$0.002} & \textbf{0.538$\pm$0.002} & \textbf{0.562$\pm$0.002}
\\

& MWNet
& 0.658$\pm$0.002  & \underline{0.818$\pm$0.005} & \underline{0.926$\pm$0.003} & 0.475$\pm$0.001  & 0.527$\pm$0.002 & \underline{0.555$\pm$0.002}
& 0.603$\pm$0.001  & 0.771$\pm$0.001  & \underline{0.890$\pm$0.001} & 0.428$\pm$0.000  & 0.483$\pm$0.000 &0.513$\pm$0.000
\\

& MetaSGD
& 0.207$\pm$0.003  & 0.309$\pm$0.005  & 0.433$\pm$0.018 & 0.136$\pm$0.002  & 0.169$\pm$0.003 & 0.200$\pm$0.005
& 0.209$\pm$0.000  & 0.321$\pm$0.003  & 0.451$\pm$0.001 & 0.137$\pm$0.000  & 0.174$\pm$0.001 & 0.206$\pm$0.002
\\

& S$^2$Meta
& 0.393$\pm$0.005  & 0.525$\pm$0.007  & 0.654$\pm$0.113 & 0.277$\pm$0.003  & 0.320$\pm$0.004 & 0.351$\pm$0.086
& 0.208$\pm$0.001  & 0.323$\pm$0.001  & 0.471$\pm$0.002 & 0.137$\pm$0.000  & 0.174$\pm$0.001 & 0.211$\pm$0.001
\\

& SML
& N/A  & N/A  & N/A  & N/A   & N/A   & N/A   
& 0.478$\pm$0.000  & 0.614$\pm$0.000 &0.720$\pm$0.000 & 0.338$\pm$0.000  & 0.382$\pm$0.000 & 0.409$\pm$0.000
\\[0.5ex]\cline{2-14}\rule{0pt}{2ex}

& \algname{}
& \textbf{0.779$\pm$0.010}  & \textbf{0.923$\pm$0.003}  & \textbf{0.980$\pm$0.006} & \textbf{0.563$\pm$0.027}  & \textbf{0.610$\pm$0.024} & \textbf{0.624$\pm$0.020}
& \underline{0.619$\pm$0.017}  & \underline{0.779$\pm$0.016}  & 0.886$\pm$0.011 & \underline{0.439$\pm$0.014} & \underline{0.491$\pm$0.014} & \underline{0.519$\pm$0.012}
\\
\bottomrule
\end{tabular}
}
\vspace*{-0.25cm}
\caption{Overall online recommendation performance.\protect\footnotemark~The average of five executions with the standard error are reported. The best results are marked in bold, and the second best results are underlined.}
\vspace*{-0.45cm}
\label{Tab:online_update_performance_comparison}
\end{table*}

\begin{figure}[htb!]
\centering
\begin{subfigure}[t!]{0.47\textwidth}
    \includegraphics[width=\textwidth]{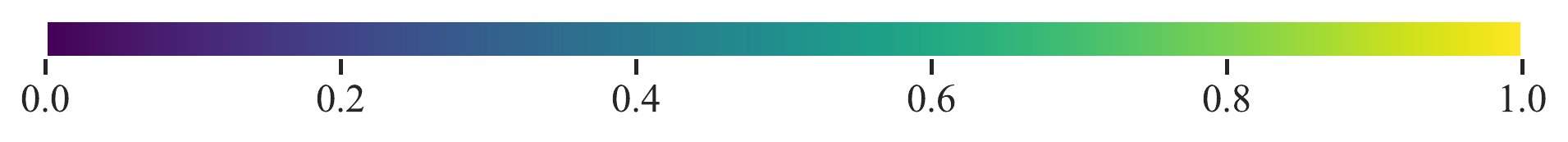}
\end{subfigure}
\begin{subfigure}[t!]{0.47\textwidth}
\includegraphics[width=0.49\textwidth]{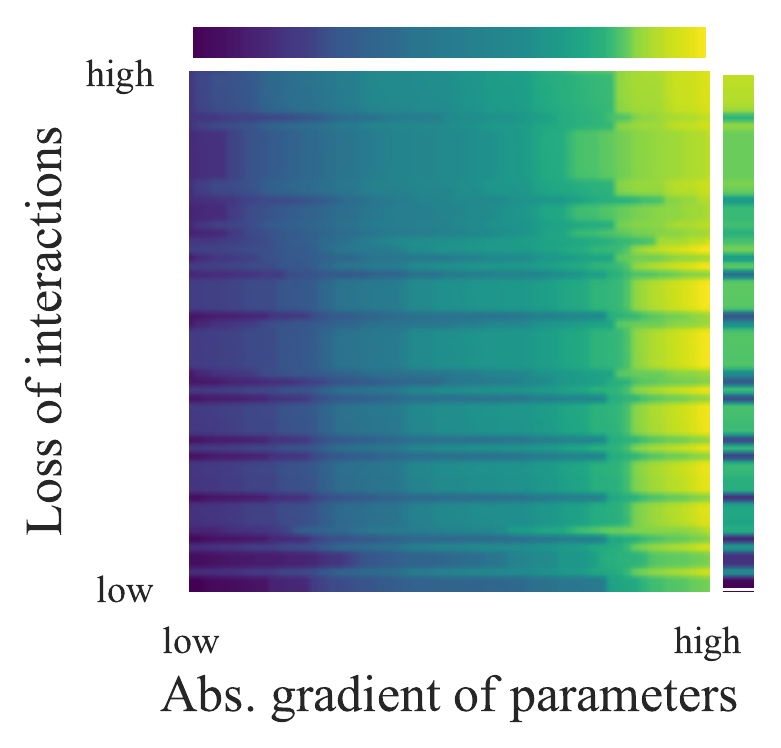}
\includegraphics[width=0.49\textwidth]{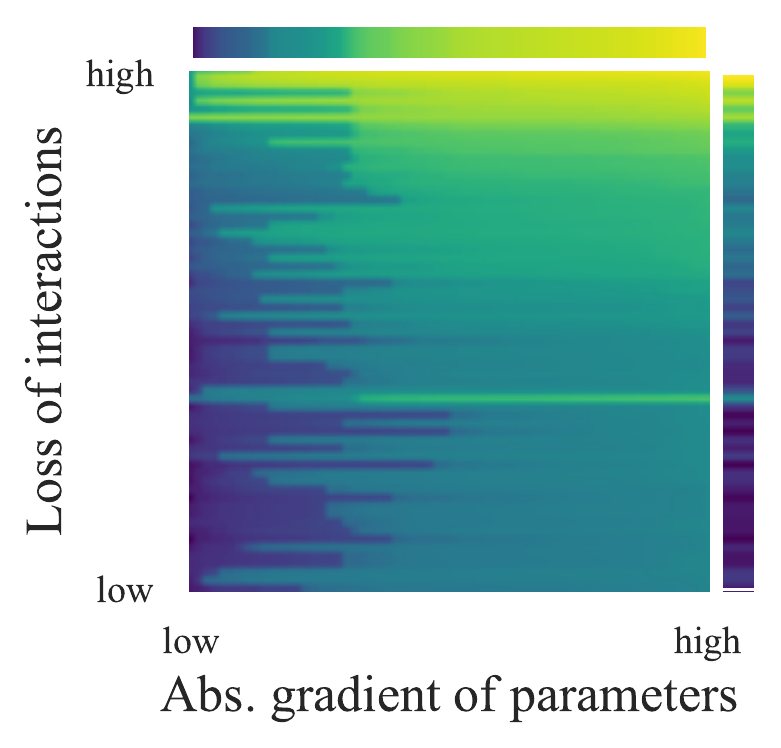}
\end{subfigure}
\text{\small \hspace{0.7cm} (a) NCF on Adressa. \hspace{1.4cm} (b) BPR on Adressa.}
\begin{subfigure}[t!]{0.47\textwidth}
\includegraphics[width=0.49\textwidth]{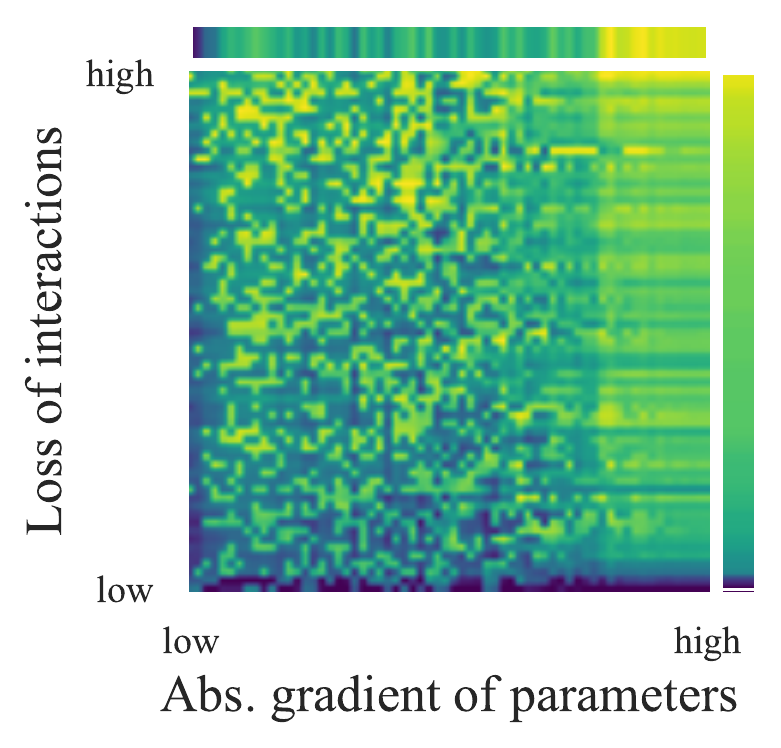}
\includegraphics[width=0.49\textwidth]{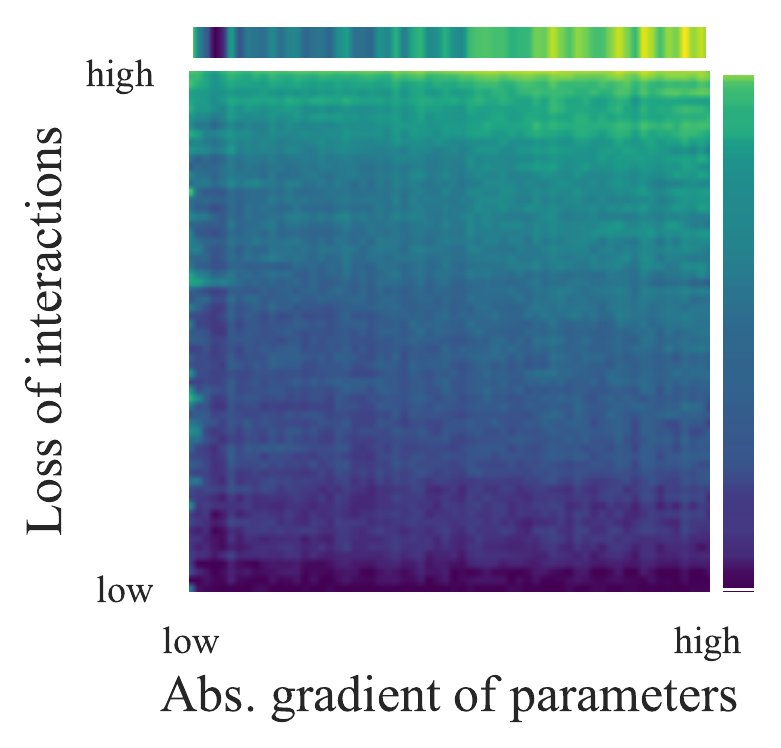}
\end{subfigure}
\text{\small \hspace{0.8cm} (c) NCF on Yelp. \hspace{1.7cm} (d) BPR on Yelp. \hspace{0.2cm}}
\vspace{-0.2cm}
\caption{Learning rates obtained by \algname{} for interaction-parameter pairs. \algname{} tends to provide higher learning rates for the pairs with large losses and gradients.}
\label{fig:MeLON_learning_rate_visualization}
\vspace*{-0.6cm}
\end{figure}

\begin{figure*}[htb!]
\begin{center}
\begin{subfigure}[t!]{0.75\textwidth}
\includegraphics[width=\textwidth]{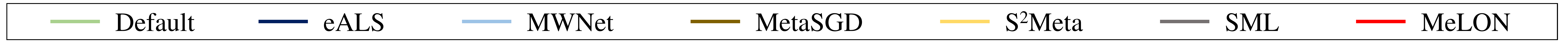}
\end{subfigure}
\end{center}
\begin{subfigure}[t!]{\textwidth}
\includegraphics[width=0.25\textwidth]{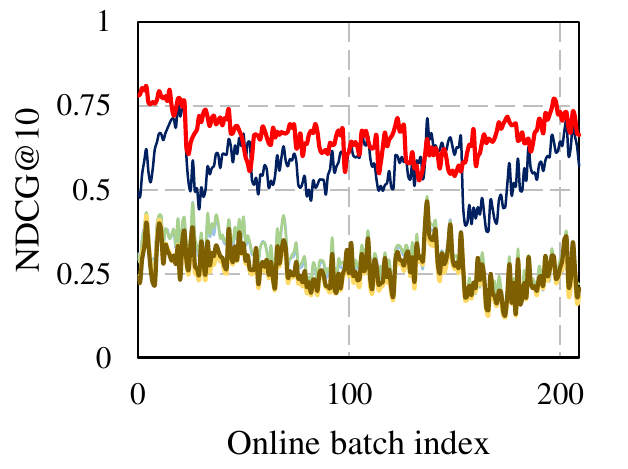}
\hspace*{-0.2cm}
\includegraphics[width=0.25\textwidth]{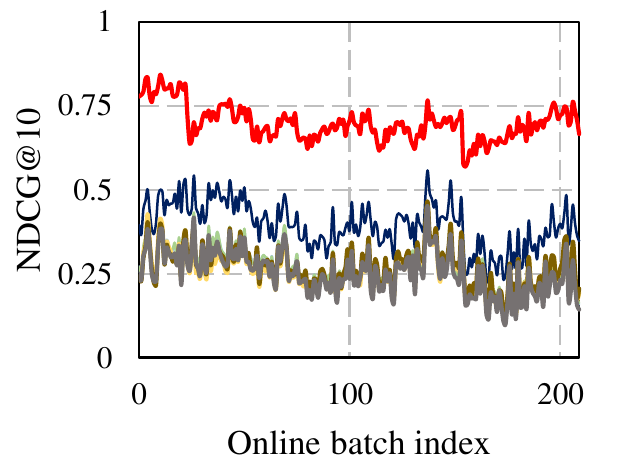}
\hspace*{-0.2cm}
\includegraphics[width=0.25\textwidth]{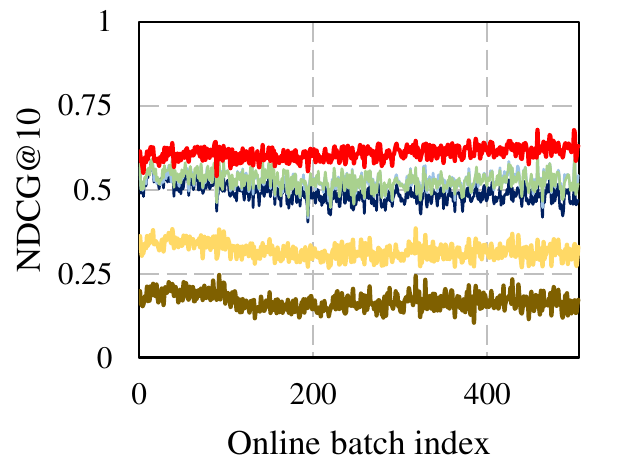}
\hspace*{-0.2cm}
\includegraphics[width=0.25\textwidth]{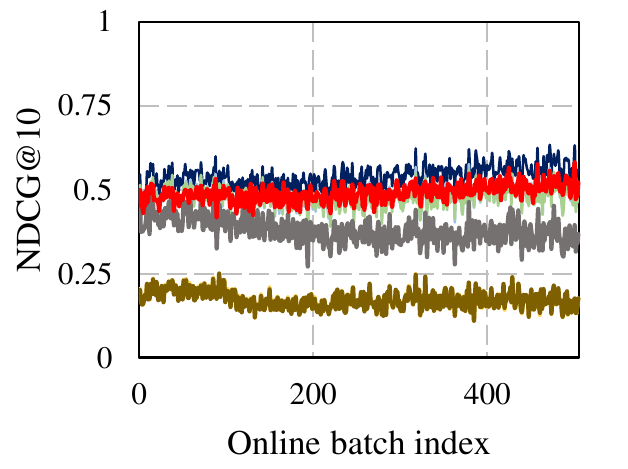}
\end{subfigure}
\\{\color{white} hh \vspace*{-0.3cm}} \\
\hspace*{1.08cm} {\small (a) NCF on Adressa.} \hspace*{1.7cm} {\small (b) BPR on Adressa.} \hspace*{1.95cm} {\small (c) NCF on Yelp.} \hspace*{1.95cm} {\small (d) BPR on Yelp.}
\vspace{-0.2cm}
\caption{Recommendation performance over each online batch in the Adressa and Yelp datasets.}
\label{fig:online_update_performance}
\vspace*{-0.4cm}
\end{figure*}

\subsubsection{Evaluation Metrics.}
We used two widely-used evaluation metrics, hit rate\,(HR) and normalized discounted cumulative gain\,(NDCG). 
Given a recommendation list, HR measures the rate of true user-interacted items in the list, while NDCG additionally considers the ranking of the user-interacted items.
The two metrics were calculated for top$@5$, top$@10$, and top$@20$ items, respectively. 
For each mini-batch on prequential evaluation, a recommender estimates the rank of each user's $1$ interacted item and randomly-sampled $99$ non-interacted items, and this technique is widely used in the literature\,\cite{he2016fast, du2019sequential} because it is time-consuming to rank all non-interacted items.

Please see Section B of the supplementary material for more details of the experiment settings.

\footnotetext{SML cannot be implemented on NCF because the transfer network of SML is intended to work on embedding parameters which do not exist in NCF.}

\subsection{Overall Performance Comparison}

Table \ref{Tab:online_update_performance_comparison} shows the top-$k$ recommendation performance with varying update strategies for the three datasets. 
Overall, \algname{} greatly boosts the recommendation performance compared with the other update strategies in general. It outperforms the other update strategies with NCF in terms of HR$@5$ by up to 29.9\%, 10.9\%, and 18.2\% in Adressa, Amazon, and Yelp, respectively.
This benefit is attributed to the two-directional flexibility of \algname{}, which successfully adapts learning rates on each interaction-parameter pair. That is, \algname{} considers the importance of the user-item interaction as well as the role of the parameter, while the compared strategies consider only either of them.

\smallskip
\noindent\textbf{Flexibility Gap.} 
Figure \ref{fig:MeLON_learning_rate_visualization} displays the learning rate matrix of all interaction-parameter pairs in \algname{} when trained on Adressa and Yelp. Here, each square matrix displays the two-directional learning rates\,($\boldsymbol{W}^{2D}$) for all interaction-parameter pairs. On the other hand, the top and right bars are the averages along one axis, which can be considered as the one-directional learning rates\,($\boldsymbol{W}^{1D}$).
The learning rates in the square matrix are flexibly determined for each interaction-parameter pair, and high learning rates are assigned when the loss or gradient is high. A row or column does \emph{not} stick to the same learning rate, and this visualization clearly demonstrates the necessity of the two-directional flexibility.
The gap between the learning rate in $\boldsymbol{W}^{2D}$ and that in $\boldsymbol{W}^{1D}$ in Figure \ref{fig:MeLON_learning_rate_visualization} is directly related to how quickly a recommender adapts to up-to-date user interests, which in turn leads to the performance difference between \algname{} and the previous update strategies.

\smallskip
\noindent\textbf{In-Depth Comparison.}
We provide interesting observations for the online update strategies: 
{
\vspace*{-0.1cm}
\begin{itemize}[leftmargin=12pt] 
\item Importance reweighting works well for datasets where a user's interest changes slowly (e.g., Yelp);
\item Meta-optimization works well for datasets where a user's interest changes quickly (e.g., Adressa);
\item \algname{} works well for both types of datasets.
\end{itemize}
\vspace*{-0.1cm}
}
Specifically, in terms of HR$@20$, an importance reweighting strategy, MWNet, enhances the recommendation performance in Yelp, but shows worse performance in Adressa than Default. In contrast, an opposite trend is observed for the two meta-optimization strategies, MetaSGD and S$^2$Meta.
Thus, we conjecture that, for time-sensitive user interest, such as news in Adressa, it is more important to focus on the parameter roles, which could be associated with the topics in this dataset. On the other hand, for time-insensitive user interest, such as places in Yelp, it would be better to focus on the interaction itself. This claim can be further supported by the different trends in Figure \ref{fig:MeLON_learning_rate_visualization}, where horizontal\,(i.e., parameter-wise) lines are more visible in the Adressa dataset, but vertical\,(i.e., interaction-wise) lines are more visible in the Yelp dataset.

\subsection{Performance Trend over Time} 
Figure \ref{fig:online_update_performance} shows the NDCG$@10$ performance trend of seven update strategies over each online test batch of the Adressa and Yelp datasets.
Overall, only \algname{} consistently adheres to the highest performance (close to the highest in Figure \ref{fig:online_update_performance}(d)) compared with other update strategies during the entire test period. The performance gap between \algname{} and the others widens especially in Addressa because its news data becomes quickly outdated and needs more aggressive adaptation for better recommendation.
In this regard, eALS also shows better performance than others since it always assigns high weights for new user-item interactions.
On the other hand, in the Yelp dataset where the user's interest may not change quickly, all the update strategies show small performance fluctuations. 


{
\begin{table}[t!]
\center
\scriptsize
\begin{tabular}{X{1.3cm} X{1.8cm} X{1.8cm} X{1.8cm} }
\toprule
 Dataset           & \algname{}$_I$      & \algname{}$_P$   & {\bf \algname{}}        \\ \midrule
 Adressa           & 0.293 $\pm$ 0.003     & 0.278 $\pm$ 0.015  & {\bf 0.877 $\pm$ 0.004}   \\
 Amazon            & 0.248 $\pm$ 0.002     & 0.269 $\pm$ 0.029  & {\bf 0.363 $\pm$ 0.016}   \\
 Yelp              & 0.605 $\pm$ 0.001     & 0.208 $\pm$ 0.000  & {\bf 0.619 $\pm$ 0.017}   \\
\bottomrule
\end{tabular}
\vspace*{-0.2cm}
\caption{Ablation on \algname{} components. HR@5 and its standard error for BPR are reported. \algname{}$_I$ and \algname{}$_P$ are the variants with only importance reweighting and meta-optimization, respectively.}
\vspace*{-0.4cm}
\label{Tab:MeLON_ablation}
\end{table}
}
\subsection{Ablation Study on Two-Directional Flexibility}
We conduct an ablation study to examine the two-directional flexibility of \algname{} by using its two variants with \emph{partial} flexibility:
(i) \algname{}$_I$, an importance reweighting variant without parameter-wise inputs and (ii) \algname{}$_P$, a meta-optimization variant without interaction-wise inputs.
Table \ref{Tab:MeLON_ablation} shows the performance of the two variants along with the original \algname{} on the three datasets. Of course, \algname{} is far better than the variants.
In addition, the performance of \algname{}$_I$ is similar to those of the importance reweighting strategies (e.g., MWNet) in Table \ref{Tab:online_update_performance_comparison}, while \algname{}$_P$ shows the results similar to the meta-optimization strategies (e.g., S$^2$Meta). Therefore, the power of \algname{} is attained when the two-directional flexibility is accompanied.

\subsection{Elapsed Time for Online Update}
Table \ref{Tab:online_update_elapsed_time} shows the average elapsed time of the seven update strategies per online batch update.
Overall, all the update strategies including \algname{} show affordable update time except S$^2$Meta which consumes even seconds in Yelp and Amazon. 
That is, 
\algname{} is still capable of handling multiple recommender updates within a second, which is fast enough for online recommender training.
The speed of \algname{} is improved by its \emph{selective} parameter update; given a user-item interaction, \algname{} updates only the parameters involved with the recommender's computation for the interaction.
This technique helps \algname{} maintain its competitive update speed, despite the use of a meta-model which is believed to be time-consuming.

\begin{table}[t]
\scriptsize
\resizebox{0.47\textwidth}{!}{
\begin{tabular}{c | c | c c c c c c |c}\toprule
Dataset & \!\! Model\!\!  & \!\!\!\! {Default}\!\!\!\!  & \!\!\!\! {eALS}\!\!\!\!          & \!\!\!\! {MWNet}\!\!\!\!  & \!\!\!\! {MetaSGD}\!\!\!\!   & \!\! {SML}\!\!    & \!\! S$^2$Meta\!\! & \!\!\!\! {MeLON}\!\!\!\!   \\ \midrule
\multirow{2}{*}{Adressa}
& NCF & 0.011 & 0.011 & 0.023 & 0.022 & N/A & 0.317 & 0.257 \\ 
& BPR & 0.011 & 0.011 & 0.012 & 0.014 & 0.163  & 0.112 & 0.076 \\ \midrule
\multirow{2}{*}{Amazon}
& NCF & 0.010 & 0.010 & 0.026 & 0.022 & N/A & 2.922 & 0.263 \\ 
& BPR & 0.009 & 0.009 & 0.014 & 0.012 & 0.130 & 1.630 & 0.072 \\ \midrule
\multirow{2}{*}{Yelp}
& NCF & 0.008 & 0.008 & 0.016 & 0.018 & N/A     & 3.345 & 0.258  \\
& BPR & 0.007 & 0.007 & 0.008 & 0.011 & 0.051 & 0.431 & 0.041  \\ \bottomrule
\end{tabular}
}
\vspace*{-0.15cm}
\caption{Elapsed time (sec) of seven update strategies per online batch.}
\vspace*{-0.6cm}
\label{Tab:online_update_elapsed_time}
\end{table}
\section{Conclusion}
\label{sec:conclusion}

In this paper, we proposed \algname{}, a meta-learning-based highly flexible update strategy for online recommender systems. 
\algname{} provides learning rates \emph{adaptively} to each parameter-interaction pair to help recommender systems be aligned with up-to-date user interests.
To this end, \algname{} first represents the meaning of a user-item interaction and the role of a parameter using a GAT and a neural mapper.
Then, the adaptation layer exploits the two representations to determine the optimal learning rate.
Extensive experiments were conducted using three real-world online service datasets, and the results confirmed the higher accuracy of \algname{} by virtue of its \emph{two}-directional flexibility as validated in the ablation study and theoretical analysis.

\section*{Acknowledgement}
This work was partly supported by the National Research Foundation of Korea\,(NRF) grant funded by the Korea government\,(Ministry of Science and ICT) (No. 2020R1A2B5B03095947) and Institute of Information \& Communications Technology Planning \& Evaluation\,(IITP) grant funded by the Korea government\,(MSIT) (No.\ 2020-0-00862, DB4DL: High-Usability and Performance In-Memory Distributed DBMS for Deep Learning).

\bibliography{8-reference}

\clearpage \appendix
\onecolumn{
\vspace*{3\baselineskip}
\begin{center}
\vspace{0.5cm}
\textbf{\LARGE Supplementary Material for Paper 2570:}
\end{center}
\begin{center}
\textbf{\LARGE Meta-Learning for Online Update of Recommender Systems}
\end{center}
\vspace{1cm}
}

\begin{multicols}{2}

\section{A.\quad Online Recommender Training}
\label{sec:algorithm}

While an online recommender can be trained on new user-item interactions with adaptive learning rates by the meta-model \algname{}, the optimality condition varies with time.
Therefore, \algname{} should  be continuously updated along with the recommender to avoid being stale.
To this end, as shown in Figure \ref{fig:MeLON_update}, for every incoming mini-batch $\mathcal{B}_t$, we first conduct two steps to train the \emph{meta-model} $\phi^{2D}$ before updating the \emph{recommender model} $\boldsymbol{\Theta}$, following the common update procedure of meta-learning\,\citeapx{ren2018learning, shu2019meta}. Note that Figure  \ref{fig:MeLON_update} is made more detailed by clarifying the two steps for $\phi^{2D}$, compared with Figure \ref{fig:recommender_update} (in the main paper).
\begin{enumerate}[leftmargin=14pt]
\item \emph{Recommender model preliminary update}: For the users and items in the new mini-batch $\mathcal{B}_t$ at each iteration, we first derive their last interactions $\mathcal{B}_t^{last}$ before the current interaction.
Then, using the \emph{current} meta-model $\phi^{2D}_t$, the parameter $\boldsymbol{\Theta_t}$ of the recommender model is updated  on the latest interactions $\mathcal{B}_t^{last}$ to create a model with $\tilde{\boldsymbol{\Theta}}$ by 
\vspace*{-0.1cm}
\begin{equation}
\tilde{\Theta}\!=\!{\Theta}_{t}\!-\!\nabla_{{\Theta}_{t}}\!\!\!\sum_{x \in \mathcal{B}_t^{last}}\!\!\!\mathcal{L}_{\Theta_t}(x)\!\cdot\!{\phi}^{2D}_{t}\Big(x,\mathcal{L}_{\Theta_t}(x),\Theta_t\Big).
\label{eq:Temporal_recommender_update}
\end{equation}
\vspace*{-0.4cm}
\item \emph{Meta-model update}: Because $\tilde{\Theta}$ obtained by Eq.\ \eqref{eq:Temporal_recommender_update} is widely known as an inspection on the efficacy of the current meta-model\,\citeapx{ren2018learning, shu2019meta}, the feedback from $\tilde{\Theta}$ is exploited to update the meta-model on the incoming mini-batch $\mathcal{B}_{t}$ by
\vspace{-0.15cm}
\begin{equation}
{\phi}^{2D}_{t+1} ={\phi}^{2D}_{t} - \eta\nabla_{{\phi}^{2D}_{t}}\!\!\sum_{x \in \mathcal{B}_t}\!\!\frac{1}{n}\mathcal{L}_{\tilde{\Theta}}(x),
\label{eq:Meta_update}
\vspace{-0.15cm}
\end{equation}
where $\eta$ is a learning rate for meta-model.

\item \emph{Recommender model update}: Finally, the parameter $\boldsymbol{\Theta}_{t}$ of the recommender model is updated using the \emph{updated} meta-model $\phi^{1D}_{t+1}$ on the  mini-batch $\mathcal{B}_{t}$ by    
\vspace{-0.15cm}
\begin{equation}
\!\!\!{\Theta}_{t+1}\!=\!{\Theta}_{t}\!-\!\nabla_{{\Theta}_{t}}\!\!\sum_{x \in \mathcal{B}_t}\!\!\mathcal{L}_{\Theta_t}(x)\cdot{\phi}^{2D}_{t+1}\Big(x,\!\mathcal{L}_{\Theta_t}(x),\Theta_t\Big).
\label{eq:Main_recommender_update}
\end{equation}
\end{enumerate}
\vspace{-0.15cm}

Note that \algname{} \emph{selectively} performs the update of recommender parameters involved in the recommender's computation for each interaction. Therefore, the required update time for \algname{} is comparable to other update strategies, such as SML and S$^2$Meta, as empirically confirmed in the evaluation results.

\begin{figure}[H]
\begin{center}
\includegraphics[width=0.47\textwidth]{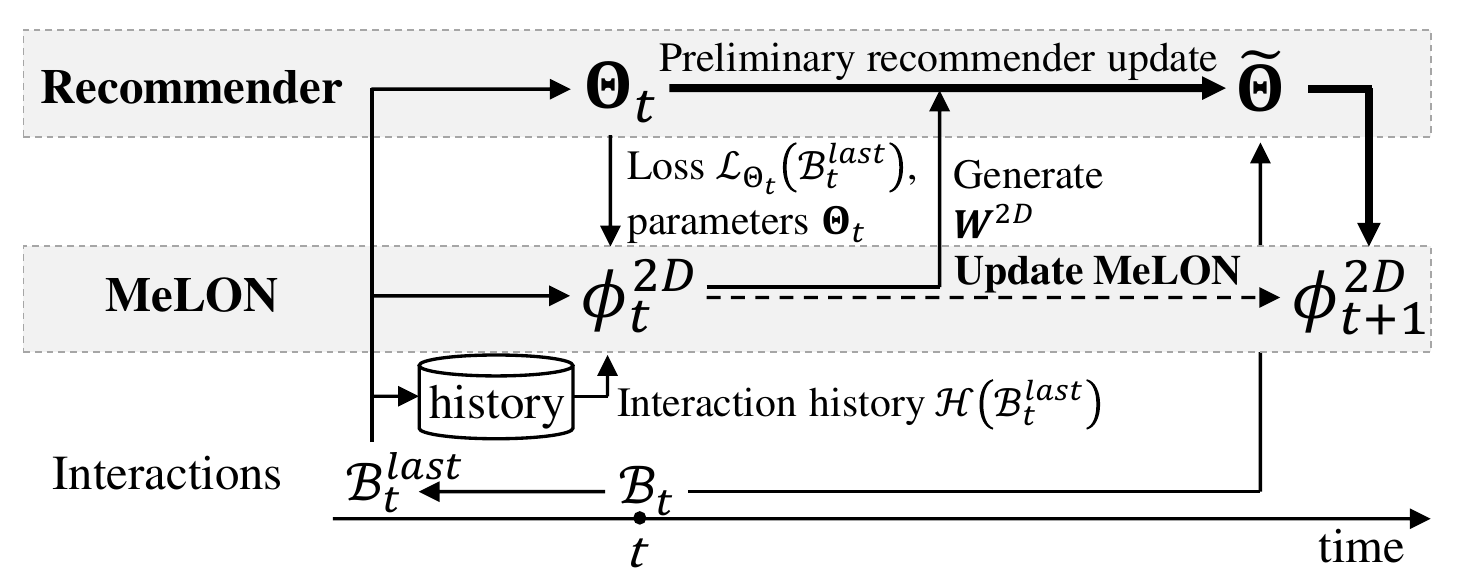}
\end{center}
\vspace*{-0.4cm}
\caption{\algname{} update procedure. For the users and items in every new mini-batch $\mathcal{B}_t$, we retrieve their last interaction $\mathcal{B}_t^{last}$ to conduct a preliminary update of the recommender model with \algname{}. Then, we update \algname{} based on the loss of the updated model $\tilde{\Theta}$ on the new mini-batch $\mathcal{B}_t$.}
\label{fig:MeLON_update}
\vspace{-0.4cm}
\end{figure}

\renewcommand{\algorithmicrequire}{\textsc{Input:}}
\renewcommand{\algorithmicensure}{\textsc{Output:}} 
\newcommand{\INDSTATE}[1][1]{\STATE\hspace{0.5#1\algorithmicindent}}

\setlength{\textfloatsep}{10pt}
\begin{algorithm}[H]
\caption{{Online Training} via \algname{}}
\label{alg:online_training}
\begin{algorithmic}[1]
\REQUIRE { ${\Theta}_t$: recommender model, $\phi^{2D}_t$: meta-model}
\STATE { $t \leftarrow 1; ~\Theta_t, \phi^{2D}_t \leftarrow $ Load pretrained models;} 
\STATE{{\bf while} user-item interactions are coming {\bf do}}
\INDSTATE[1] $\mathcal{B}_{t} \leftarrow$ Get current mini-batch;
\INDSTATE[1] \COMMENT{\textbf{Meta-model update}}
\INDSTATE[1] \COMMENT{(1) Preliminary update by Eq.\,\eqref{eq:Temporal_recommender_update}}
\INDSTATE[1] {\small  $\tilde{\Theta} ={\Theta}_{t} - \nabla_{{\Theta}_{t}}\!\!\sum_{x \in \mathcal{B}_t^{last}}\!\!\mathcal{L}_{\Theta_t}(x)\cdot{\phi}^{2D}_{t}\big(x,\mathcal{L}_{\Theta_t}(x),\Theta_t\big)$};
\INDSTATE[1] \COMMENT{(2) Meta-model update by Eq.\,\eqref{eq:Meta_update}}
\INDSTATE[1] ${\phi}^{2D}_{t+1} ={\phi}^{2D}_{t} - \eta\nabla_{{\phi}^{2D}_{t}}\!\!\sum_{x \in \mathcal{B}_t}\!\!\frac{1}{n}\mathcal{L}_{\tilde{\Theta}}(x)$;
\INDSTATE[1] \COMMENT{\textbf{Recommender update}} 
\INDSTATE[1] \COMMENT{(3) Recommender model update by Eq.\,\eqref{eq:Main_recommender_update}}
\INDSTATE[1] {\small  $\!\!{\Theta}_{t+1}\!=\!{\Theta}_{t}\!-\!\nabla_{{\Theta}_{t}}\!\!\sum_{x \in \mathcal{B}_t}\!\!\mathcal{L}_{\Theta_t}(x)\cdot{\phi}^{2D}_{t+1}\big(x,\mathcal{L}_{\Theta_t}(x),\Theta_t\big)$;}
\INDSTATE[1] {$t \leftarrow t+1;$}
\end{algorithmic}
\end{algorithm}

The online training procedure of \algname{} is described in Algorithm \ref{alg:online_training}.
When a recommender is deployed online, the algorithm conducts the three steps for every new incoming mini-batch of user-item interactions: (1) a preliminary update of the recommender model\,(Lines 5--6) on the last interactions of the users and items in the current mini-batch, (2) an update of the meta-model on a new mini-batch\,(Lines 7--8), and (3) an update of the recommender model on the new mini-batch\,(Lines 10--11). 
We additionally learn a forgetting rate for the current parameter to help quick adaptation by forgetting previous outdated information\,\citeapx{ravi2016optimization}. 
The above procedure repeats for every incoming mini-batch during online service.

Before a recommender is deployed online, both the recommender and the meta-model are typically pre-trained on the past user-item interactions in an {offline} manner. 
Differently from the online training, we first randomly sample a mini-batch $\mathcal{B}$ of user-item interactions to derive the interactions $\mathcal{B}^{last}$.
Then, in each iteration, the recommender and the meta-model are updated in the same way as in the online learning.
The model is trained for a fixed number of epochs, $100$ in our experiments.
Once the offline training completes, we can deploy the recommender and the meta-model in the online recommendation environment.

\newcolumntype{L}[1]{>{\raggedright\let\newline\\\arraybackslash\hspace{0pt}}m{#1}}
\newcolumntype{X}[1]{>{\centering\let\newline\\\arraybackslash\hspace{0pt}}p{#1}}
\newcolumntype{M}[1]{>{\centering\arraybackslash}m{#1}}

\section{B.\quad Details of Experiment Settings}
\label{sec:appendix_data}

Four reproducibility, the source code of \algname{} as well as the datasets are provided as the supplementary material. 

\vspace*{-0.2cm}
\subsection{Datasets}
The explicit user ratings in the Yelp dataset and three Amazon datasets are converted into implicit ones, following the relevant researches\,\citeapx{koren2008factorization, he2017neural};
that is, if a user rated an item, then the rating is considered as a positive user-item interaction.
For prequential evaluation on online recommendation scenarios, we follow a commonly-used approach\,\citeapx{he2016fast}; we sort the interactions in the dataset in chronological order, and divide them into three parts---offline pre-training data, online validation data, and online test data.
Online validation data is exploited to search the hyperparameter setting of the recommenders and update strategies and takes up $10\%$ of test data.
Because user-item interactions are very sparse, we preprocess the datasets, following the previous approaches\,\citeapx{he2017neural, zhang2020retrain}; for all datasets, users and items involved with less than $20$ interactions are filtered out.

Table \ref{Tab:dataset} summarizes the profiles of the five datasets used in the experiments, where the details are as follows.

\textbf{Adressa}
news dataset\,\citeapx{gulla2017adressa} contains user interactions with news articles for one week.
We use the first $95\%$ of data as offline pre-training data, the next $0.5\%$ as online validation data, and the last $4.5\%$ as online test data.

\textbf{Amazon}
review dataset\,\citeapx{ni2019justifying} contains user reviews for the products purchased in Amazon.
Among various categories, we adopt three frequently-used categories, \emph{Book}\,\citeapx{wang2019neural}, \emph{Electronics}\,\citeapx{zhou2018deep}, and \emph{Grocery and Gourmet Food}\,\citeapx{wang2020make}, which vary in the size of interactions and the number of users and items. 
Because there exists almost no overlap among the categories, we perform evaluation on each category and report the average.
Due to the difference in size, we apply different data split ratios for each category.
\begin{itemize}[noitemsep]
    \item Book: $95\%$\,(pre-training): $0.5\%$\,(validation): $4.5\%$\,(test)
    \item Electronics: $90\%$\,(pre-training):$1\%$\,(validation):$9\%$\,(test) 
    \item Grocery: $80\%$\,(pre-training): $2\%$\,(validation): $18\%$\,(test)
\end{itemize}

\textbf{Yelp}
review dataset\footnote{https://www.kaggle.com/yelp-dataset/yelp-dataset}
 contains user reviews for venues, such as bars, cafes, and restaurants.
We use the first $95\%$ of data as offline pre-training data, the next $0.5\%$ as online validation data, and the last $4.5\%$ as online test data.

\begin{table}[H]
\center
\scriptsize
\vspace*{-0.15cm}
\begin{tabular}{crrrc}
\toprule
 Dataset           &  Users      &  Items   &  Interactions & Cold user \\ \midrule
 Adressa           & 29,589     & 1,457   & 1,191,114  & \,\,\,0\% \\  
 Amazon (Book)              & 80,464     & 98,663  & 3,357,109  &      18\% \\ 
 Amazon (Electronics)       & 9,316      & 17,935  & 238,458    & \,\,\,1\% \\ 
 Amazon (Grocery)           & 1,233      & 1,433   & 29,782     & \,\,\,1\% \\ 
 Yelp              & 60,543     & 74,249  & 2,880,520  & \,\,\,0\% \\  
\bottomrule
\end{tabular}
\caption{Summary of the five real-world datasets. A cold user refers to the users who do not exist in the pre-training data but in the online test data.}
\label{Tab:dataset}
\vspace*{-0.1cm}
\end{table}

\subsection{Recommender Baseline Models}
For online recommenders, we use two famous personalized recommender algorithms: BPR\,\citeapx{koren2009matrix, rendle2012bpr} and NCF\,\citeapx{he2017neural}.
\begin{itemize}[leftmargin=13pt]
    \item BPR: To handle implicit feedback, the Bayesian personalized ranking\,(BPR) uses the identifiers of a user and an item to estimate the user's interest on the item by multiplying the user embedding vector $\mathbf{e}_u$ and the item embedding vector $\mathbf{e}_i$.
    \item NCF: Neural collaborative filtering\,(NCF) maintains a generalized BPR and a multi-layer perceptron that have user and item vectors respectively. The results of these two components are later fused by a neural layer to predict a user's interest on an item.
\end{itemize}

To train these recommender algorithms based on implicit feedback data, we employ a ranking loss\,\citeapx{rendle2012bpr};
for a positive item in a user-item interaction, we randomly sample another negative item that the user has not interacted before, and train a recommender algorithm to prioritize the positive item over the negative item.

\subsection{Configuration}
For fair comparison, we follow the optimal hyperparameter settings of the baselines as reported in the original papers, and optimize uncharted ones using HR$@5$ on validation data.
All the experiments are performed with a batch size $256$ and trained for $100$ epochs.
The number of updates on the default update strategy is fixed to be 1 to align with other compared strategies.
The experiments are repeated $5$ times varying random seeds, and we report the average as well as the standard error.
For the graph attention in the first component of \algname{}, we randomly sample $10$ neighbors per target user or item. Besides, for the MLP which learns the parameter roles, the number of hidden layers ($L$) is set to be $2$. 
To optimize a recommender under the default and sample reweighting strategies, we use  Adam\,\citeapx{kingma2014adam} with a learning rate $\eta = 0.001$ and a weight decay $0.001$.
Note that a recommender trained with the meta-optimization strategies is optimized by a meta-model, while the meta-model is optimized with Adam during the meta-update in Eqs.\,\eqref{eq:Temporal_recommender_update} and \eqref{eq:Meta_update}.
Our implementation is written in PyTorch, and the experiments were conducted on Nvidia Titan RTX.
\bibliographystyleapx{aaai22}
\bibliographyapx{8-reference}
\end{multicols}

\end{document}